\documentclass[a4paper,11pt]{article}

\usepackage[T1]{fontenc}
\usepackage{aeguill}
\usepackage[utf8x]{inputenc}
\usepackage[english]{babel}
\usepackage{hyperref}
\usepackage{amsfonts, amsthm, amsmath}
\usepackage{tikz-qtree}

\title{Hofstadter's problem for curious readers}
\author{Pierre letouzey}
\date{\small
{\tt pierre.letouzey@inria.fr}\\
Team $\pi r^2$, PPS, Univ Paris Diderot \& INRIA Paris-Rocquencourt\\
Technical report version 1.5, Apr. 2018, \href{http://creativecommons.org/licenses/by/4.0/}{license CC-BY}
}

\begin{document}
\newtheorem{theorem}{Theorem}
\newtheorem{definition}{Definition}
\maketitle

\newcommand{\docgen}[2]{\href{http://www.pps.univ-paris-diderot.fr/~letouzey/hofstadter_g/doc/#1.html#2}{\tt #1.v}}
\newcommand{\doc}[1]{\docgen{#1}{}}
\newcommand{\doclab}[2]{\docgen{#1}{\##2}}
\newcommand{\FG}{\ensuremath{\overline{G}}}
\newcommand{\fibrest}{\ensuremath{\Sigma F_i}}
\newcommand{\flip}{\textit{flip}}
\newcommand{\depth}{\textit{depth}}

\section{Introduction}
This document summarizes the proofs made during a Coq development in
Summer 2015. This development investigates the function $G$ introduced
by Hofstadter in his famous ``Gödel, Escher, Bach'' book\cite{GEB},
as well as a related infinite tree. The left/right flipped variant
$\FG$ of this $G$ tree has also been studied here, following
Hofstadter's ``problem for the curious reader''.
The initial $G$ function is refered as sequence A005206 in
OEIS\cite{OEIS-G}, while the flipped version $\FG$ is the sequence
A123070 \cite{OEIS-FG}.

The detailed and machine-checked proofs can be found in the files
of this development\footnote{See
\url{http://www.pps.univ-paris-diderot.fr/~letouzey/hofstadter_g}}
and can be re-checked by running Coq \cite{Coq} version 8.4 on it.
No prior knowledge of Coq is assumed here, on the contrary this
document has rather been a ``Coq-to-English'' translation
exercise for the author. Nonetheless, some proofs given in this
document are still quite sketchy: in this case, the interested
reader is encouraged to consult the Coq files given as references.

\section{Prior art and contributions}
Most of the results proved here were already mentioned elsewhere,
e.g. on the OEIS pages for $G$ and $\FG$. These results were certified
in Coq without consulting the existing proofs in the literature, so
the proofs presented here might still be improved. To the best of
my knowledge, the main novelties of this development are:
\begin{enumerate}
\item
  A proof of the recursive
equation for $\FG$ which is currently mentioned as a conjecture in OEIS:

$$\forall n>3, \FG(n) = n+1 - \FG(1+\FG(n-1))$$

\item
  The statement and proof of another equation for $\FG$:

$$\forall n>3, \FG(\FG(n)) + \FG(n-1) = n$$

Although simpler than the previous one, this equation isn't
enough to characterize $\FG$ unless an extra monotonocity
condition on $\FG$ is assumed.

\item
  The statement and proof of a result comparing $\FG$ and $G$:
for all $n$, $\FG(n)$ is either $G(n)+1$ or $G(n)$, depending
of the shape of the Zeckendorf decomposition of $n$ as a sum
of Fibonacci numbers. More precisely, when this decomposition
of $n$ starts with $F_3$ and has a second smallest term $F_k$
with $k$ odd, then $\FG(n)=G(n)+1$, otherwise $\FG(n)=G(n)$.
Moreover these
specific numbers where $\FG$ and $G$ differ are separated by either
5 or 8.

\item The studies of $G$ and $\FG$ ``derivatives''
  ($\Delta G(n) = G(n+1)-G(n)$ and similarly for $\Delta\FG$),
  leading to yet another characterization of $G$ and $\FG$.

\end{enumerate}

\section{Coq jargon}
For readers without prior knowledge
of Coq but curious enough to have a look at the actual files
of this development, here comes a few words about the Coq syntax.
\begin{itemize}
\item By default, we're manipulating natural numbers only, corresponding
 to Coq type {\tt nat}. Only exception: in file \doc{Phi},
 we switch to real numbers to represent the golden ratio.
\item The symbol {\tt S} denotes the successor of natural numbers.
  Hence {\tt S($2$)} is the same as $3$.
\item As in many modern functional languages such as OCaml or Haskell,
  the usage is to skip parenthesis whenever possible, and
  in particular around atomic arguments of functions. Hence
  {\tt S(x)} will rather be written {\tt S x}.
\item The symbol {\tt pred} is the predecessor for natural numbers.
  In Coq, all functions must be total, and {\tt pred $0$} has been
  chosen to be equal to $0$. Similarly, the subtraction on
  natural numbers is rounded: {\tt $3$ - $5$ = $0$}.
\item Coq allows to define custom predicates to express various
  properties of numbers, lists, etc. These custom predicates are
  introduced by the keyword {\tt Inductive}, followed by some
  ``introduction'' rules for the new predicate. In this development
  we use capitalized names for these predicates (e.g. {\tt Delta},
  {\tt Low}, ...).
\item Coq also accepts new definition of recursive functions via
  the command {\tt Fixpoint}. But these definitions should satisfy
  some criterion to guarantee that the new function is well-defined
  and total. So in this development, we also need sometimes to
  define new functions via explicit justification of termination,
  see for instance {\tt norm} in \doclab{Fib}{norm}\ or {\tt g\_spec} in
  \doclab{FunG}{g\_spec}.
\end{itemize}

\section{Fibonacci numbers and decompositions}

This section corresponds to file \doc{Fib}.

\paragraph{Convention.} We use here\footnote{
Please note that earlier versions of this Coq development
and of this document were using a non-standard definition
of the Fibonacci sequence, where the leading 0 was omitted.
This non-standard definition $F_0=F_1=1, F_2=2, ...$ was meant
to ease some Coq proofs, but was actually providing little gain,
while being quite confusing for external readers.}
the following definition of the Fibonacci numbers $F_n$ :
$$F_0 = 0 $$
$$F_1 = 1 $$
$$\forall n,~~ F_{n+2} = F_{n}+F_{n+1}$$
This definition is standard, see for instance OEIS's sequence A000045
\cite{OEIS-Fib}.
In Coq, these $F_n$ numbers correspond to the {\tt fib}
function, and we start by proving a few basic properties :
strict positivity except for $F_0$, monotony, strict monotony
above 2, etc. 
We also prove the theorem {\tt fib\_inv} which states 
that any positive number can by bounded by
consecutive Fibonacci numbers :

\begin{theorem}[Fibonacci inverse]\label{fibinv}
$\forall n, \exists k, F_k \le n < F_{k+1}$.
\end{theorem}

\subsection{Fibonacci decompositions} The rest of the file \doc{Fib}\
deals with the decompositions of numbers as sums of Fibonacci
numbers.
\begin{definition}
 A decomposition $n = \fibrest$ is said to be \emph{canonical} if:
\begin{itemize}
\item[(a)] $F_0$ and $F_1$ do not appear in the decomposition
\item[(b)] A Fibonacci number appears at most once in the decomposition
\item[(c)] Fibonacci numbers in the decomposition aren't consecutive
\end{itemize}
A decomposition is said to be \emph{relaxed} if at least conditions
(a) and (b) hold. 
\end{definition}
For instance, $11$ admits the canonical decomposition
$F_4+F_6 = 3 + 8$, but also the relaxed decomposition $F_2+F_3+F_4+F_5
= 1 + 2 + 3 + 5$.
On a technical level, we represented in Coq the decompositions
as sorted lists of ranks, and we used a predicate {\tt Delta $p$ $l$}
to express that any element in the list $l$ exceeds the previous
element by at least $p$. We hence use $p=2$ for canonical
decompositions, and $p=1$ for relaxed ones. See file
\doc{DeltaList}\ for the definition and properties of this
predicate.

Then we proved Zeckendorf's theorem (actually discovered earlier by Lekkerkerker):

\begin{theorem}[Zeckendorf]\label{zeck}
Any natural number has a unique canonical decomposition.
\end{theorem}

\begin{proof}
The proof of this theorem is quite standard : for $n=0$ we take
the empty decomposition (and not the useless $F_0$), while for
$n>0$, we take first the highest $k$ such that $F_k$ is less or equal to $n$
(cf {\tt fib\_inv} below),
and continue recursively with $n-F_k$. We cannot obtain this way
two consecutive $F_{k+1}$ and $F_k$, otherwise $F_{k+2}$ could have
been used in the first place. This process will never use $F_1$,
since $F_2$ has the same value for a higher rank, and will hence
be preferred.
For proving the uniqueness, the key ingredient is that this kind of
sum cannot exceed the next Fibonacci number. For instance
$F_2+F_4+F_6 = 1+3+8 = 12 = F_7 - 1$.
\end{proof}

\begin{definition}
For a non-empty decomposition, its \emph{lowest rank} is the rank of
the lowest Fibonacci number in this decomposition. By extension,
the lowest rank $low(n)$ of a number $n\neq 0$ is the lowest
rank of its (unique) canonical decomposition.
\end{definition}
For instance $low(11)=4$. In Coq, $low(n)=k$ is written
via the predicate {\tt Low 2 $n$ $k$}. Note that this notion is
in fact more general : {\tt Low 1 $n$ $k$} says that $n$ can be
decomposed via a relaxed decomposition of lowest rank $k$.

Interestingly, a relaxed decomposition can be transformed into
a canonical one (see Coq function {\tt norm}):

\begin{theorem}[normalization]\label{norm}
Consider a relaxed decomposition of $n$, made of $k$ terms and
whose lowest rank is $r$. We can build a canonical decomposition
of $n$, made of no more than $k$ terms, and whose lowest rank can
be written $r+2p$ with $p\ge 0$.  
\end{theorem}
\begin{proof}
We simply eliminate the highest consecutive Fibonacci numbers (if any)
by summing them : $F_m+F_{m+1} \to F_{m+2}$, and repeat this
process. By dealing with highest numbers first, we avoid the apparition of
duplicated terms in the decomposition: $F_{m+2}$ wasn't already
in the decomposition, otherwise we would have considered
$F_{m+1}+F_{m+2}$ instead. Hence all the obtained decompositions
during the process are correct relaxed decompositions. The number of
terms is decreasing by 1 at each step, so the process is guaranteed to
terminate, and will necessarily stops on a decomposition with no
consecutive Fibonacci numbers : we obtain indeed the canonical
decomposition of $n$. And finally, it is easy to see that the
lowest rank is either left intact or raised by 2 at each step
of the process.
\end{proof}

\subsection{Classifications of Fibonacci decompositions.}
The end of \doc{Fib}\ deals with classifications of numbers
according to the lowest rank of their canonical decomposition.
In particular, this lowest rank could be 2, 3 or more. It will
also be interesting to distinguish between lowest ranks that
are even or odd. These kind of classifications and their
properties will be heavily used
during theorem \ref{comp-fg-g} which compares $\FG$ and $G$.
For instance:
\begin{theorem}[decomposition of successor]\label{fibsucc}
\noindent
\begin{enumerate}
\item $low(n) = 2$ implies $low(n+1)$ is odd.
\item $low(n) = 3$ implies $low(n+1)$ is even and different from 2.
\item $low(n) > 3$ implies $low(n+1) = 2$.
\end{enumerate}
\end{theorem}
\begin{proof}
Let $r$ be $low(n)$. We can write $n = F_r + \fibrest$ in a canonical
way.
\begin{enumerate}
\item If $r = 2$, then $n+1 = F_3 + \fibrest$ is a relaxed decomposition,
  and we conclude thanks to the previous normalization theorem.
\item If $r = 3$, then $n+1 = F_4 + \fibrest$ is a relaxed decomposition,
  and we conclude similarly.
\item If $r > 3$, then $n+1 = F_2 + F_r + \fibrest$ is directly a
  canonical decomposition.
\end{enumerate}
\end{proof}

We also studied the decomposition of the predecessor
$n-1$ in function of $n$. This decomposition depends of the parity of
the lowest rank in $n$, since for $r\neq 0$ we have:

$$ F_{2r} - 1 = F_3 + ... + F_{2r-1}$$
$$ F_{2r+1} - 1 = F_2 + F_4 + ... + F_{2r}$$

Hence:

\begin{theorem}[decomposition of predecessor]\label{fibpred}
\noindent
\begin{enumerate}
\item If $low(n)$ is odd then $low(n-1) = 2$.
\item If $low(n)$ is even and different from 2, then $low(n-1) = 3$.
\item For $n>1$, if $low(n)=2$ then $low(n-1)>3$.
\end{enumerate}
\end{theorem}
\begin{proof}
\noindent
\begin{enumerate}
\item Let the canonical decomposition of $n$ be $F_{2r+1} +
  \fibrest$. Then we have $n-1 = F_2 + ... + F_{2r} + \fibrest$, which
 is also a canonical decomposition. Moreover $r\neq 0$ since $F_1$
 isn't allowed in decompositions. Hence the decomposition of $n-1$
 contains at least the term $F_2$, and $low(n-1)=2$.
\item When $low(n) = 2r$ with $r\neq 0$, we decompose
 similarly $F_{2r}-1$, leading to a lowest rank 3 for $n-1$.
\item Finally, when $n = F_2 + \fibrest$, then the canonical decomposition
  of $n-1$ is directly the rest of the decomposition of $n$, which
  isn't empty (thanks to the condition $n>1$) and
  cannot starts by $F_2$ nor $F_3$ (for canonicity reasons).
\end{enumerate}
\end{proof}

\subsection{Subdivision of decompositions starting by $F_3$}

When $low(n)=3$ and $n>2$, the canonical decomposition of $n$
admits at least a second-lowest term: $n = F_3 + F_k + ...$ and we
will sometimes
need to consider the parity of this second-lowest rank.

\begin{definition}
\noindent
\begin{itemize}
\item A number $n$ is said \emph{3-even} when its canonical decomposition
  is of the form $F_3 + F_{2p} + ...$.
\item A number $n$ is said \emph{3-odd} when its canonical decomposition
  is of the form $F_3 + F_{2p+1} + ...$
\end{itemize}
\end{definition}
In Coq, this corresponds to the {\tt ThreeEven} and {\tt ThreeOdd}
predicates.
In fact, the 3-odd numbers will precisely be the locations where
$G$ and $\FG$ differs (see theorem \ref{comp-fg-g}). The first of these 3-odd numbers is
$7 = F_3+F_5$, followed by $15 = F_3 + F_7$ and $20 = F_3+F_5+F_7$.

A few properties of 3-even and 3-odd numbers:
\begin{theorem}\label{threeevenodd}
\noindent
\begin{enumerate}
\item A number $n$ is 3-odd if and only if it admits
a \emph{relaxed} decomposition of the form $F_3 + F_{2p+1}+...$.
\item A number $n$ is 3-odd if and only if $low(n)=3$ and
$low(n-2)$ is odd.
\item A number $n$ is 3-even if and only if $low(n)=3$ and
$low(n-2)$ is even.
\item When $low(n)$ is even and at least 6, then $n-1$ is 3-odd.
\item Two consecutive 3-odd numbers are always apart by 5 or 8.
\end{enumerate}
\end{theorem}

\begin{proof}
\noindent
\begin{enumerate}
\item
  The left-to-right implication is obvious, since a canonical
  decomposition can in particular be seen as a relaxed one. Suppose
  now the existence of such a relaxed decomposition. During its
  normalization, the $F_3$ will necessarily be left intact, and the
  term $F_{2p+1}$ can grow, but only by steps of 2.
\item Once again, the left-to-right implication is
  obvious. Conversely, we start with the canonical decomposition of
  $n-2=F_{2p+1}+\fibrest$. If $p>1$, this leads to the desired canonical
  decomposition of $n$ as $F_2+F_{2p+1}+\fibrest$. And the case $p\le 1$ is
  impossible: $p=0$ leads to $F_1$ being part of a relaxed
  decomposition, and if we assume $p=1$, then we could turn the
  decomposition of $n-2$ into a canonical decomposition of $n-4$
  whose lowest rank is at least 5. This gives us a relaxed
  decomposition of $n$ of the form $F_2+F_4+\fibrest$. After normalization,
  we would obtain that $low(n)$ is even, which is contradictory with
  $low(n)=3$.
\item We proceed similarly for $low(n)=3$ and $low(n-2)$ even implies
  that $n$ is 3-even. First we have a canonical decomposition of
  $n-2 = F_{2p} + \fibrest$. Then $p=0$ is illegal, and $p=1$
  would give a relaxed
  decomposition of $n$ starting by 3, hence $low(n)$ even, impossible.
  And $p>1$ allows to write $n = F_3+F_{2p}+\fibrest$ in a canonical way.
\item When $n = F_{2p}+\fibrest$ with $p\ge 3$, then
  $n-1 = F_3+F_5+...+F_{2p-1}+\fibrest$
  and the condition on $p$ ensures that the terms $F_3$ and $F_5$ are
  indeed present.
\item If $n$ is 3-odd, it could be written $F_3+F_{2p+1}+\fibrest$.
  When $p>2$, then $n+5 = F_3+F_5+F_{2p+1}+\fibrest$ is also 3-odd.
  When $p=2$, then $n+8$ is also 3-odd. For justifying that,
  we need to consider the third term in the 
  canonical decomposition (if any).
  \begin{itemize}
  \item Either $n$ is of the form $F_3+F_5+F_7+...$
    Then $n+F_6 = F_3 + F_5 + F_8 + ...$ is a relaxed
    decomposition of $n+8$, hence $n+8$ is 3-odd (see the first
    part of the current theorem above).
  \item Either $n$ is of the form $F_3+F_5+\fibrest$
    where all terms in $\fibrest$ are strictly greater than $F_7$.
    Then $n+F_6 = F_3 + F_7 + \fibrest$ is a relaxed
    decomposition of $n+8$, hence $n+8$ is 3-odd.
  \end{itemize}
  We finish the proof by considering all intermediate numbers between
  $n$ and $n+8$ and we show that $n+5$ is the only one of them that
  might be 3-odd:
  \begin{itemize}
  \item $low(n+1)$ is even (cf. theorem \ref{fibsucc}).
  \item $n+2 = F_2+F_4+F_{2p+1}+\fibrest$, and normalizing this
    relaxed decomposition shows that $low(n+2)=2$.
  \item $n+3 = F_3+F_4+F_{2p+1}+\fibrest$, and normalizing this
    relaxed decomposition will either combine $F_4$ with some
    higher terms (leading to an odd second-lowest term and
    hence $n+3$ is 3-even) or combine $F_4$ with $F_3$ (in which
    case $low(n+3)\ge 5$).
  \item $n+4 = F_2+F_3+F_4+F_{2p+1}+\fibrest$, hence $low(n+4)$ is
    even.
  \item $n+6 = F_6+F_{2p+1}+\fibrest$ is a relaxed decomposition
   whose lowest rank is either 6 (when $p>2$) or 5 (when $p=2$).
   After normalization, we obtain $low(n+6)\ge 5$.
  \item Hence $low(n+7)=2$ by last case of theorem \ref{fibsucc}.
  \end{itemize}
\end{enumerate}
\end{proof}

\section{The $G$ function}

This section corresponds to file \doc{FunG}.

\subsection{Definition and initial study of $G$}

\begin{theorem}\label{defG}There exists a unique function
  $G:\mathbb{N}\to\mathbb{N}$ which satisfies the following
  equations:
  $$G(0) = 0$$
  $$\forall n>0,~~ G(n) = n - G(G(n-1))$$
\end{theorem}
\begin{proof}
The difficulty is the second recursive call of $G$ on $G(n-1)$.
A priori $G(n-1)$ could be arbitrary, so how could we refer to it
during a recursive definition of $G$ ? Fortunately $G(n-1)$ will
always be strictly lower than $n$, and that will allow this recursive
call. More formally, we prove by induction on $n$ the following
statement: for all $n$, there exists a sequence $G_0...G_n$
of numbers in $[0..n]$ such that $G_0=0$ and
$\forall k\in[1..n], G_k = k - G_{k'}$ where $k' = G_{k-1}$.
\begin{itemize}
\item For $n=0$, an adequate sequence is obviously $G_0 = 0$.
\item For some $n$, suppose we have already proved the existence
  of an adequate sequence $G_0...G_n$.
  In particular $G_n \in [0..n]$, hence $G_{G_n}$
  is well-defined and is also in $[0..n]$. Finally $(n+1)-G_{G_n}$
  is in $[1..(n+1)]$. We define $G_{n+1}$ to be equal to this
  $(n+1)-G_{G_n}$ and obviously keep the previous values of the
  sequence. All these values are indeed in $[0..(n+1)]$, and the
  recursive equations are satisfied up to $k=n+1$. 
\end{itemize}
All these finite sequences that extend each other
lead to an infinite sequence $(G_n)_{n\in\mathbb{N}}$ of natural
numbers, which can also be seen as a function
$G:\mathbb{N}\to\mathbb{N}$, that satisfy the desired equations
by construction.

For the uniqueness, we should first prove that any function $f$
satisfying $f(0)=0$ and our recursive equation above is such that
$\forall n, 0\le f(n)\le n$. This proof can be done by strong
induction over $n$, and a bit of upper and lower bound manipulation.
Then we could prove (still via strong induction
over $n$) that $\forall n, f(n)=g(n)$ when $f$ and $g$
are any functions satisfying our equations. 
\end{proof}

The initial values are $G(0)=0, G(1)=G(2)=1, G(3)=2, G(4)=G(5)=3$.
We can then establish some basic properties of $G$:
\begin{theorem}\label{Gprops}
\noindent
\begin{enumerate}
\item $\forall n, 0 \le G(n) \le n$.
\item $\forall n\neq 0, G(n)=G(n-1)$ implies $G(n+1)=G(n)+1$.
\item $\forall n, G(n+1)-G(n) \in \{0,1\}$.
\item $\forall n,m, n\le m$ implies $0 \le G(m)-G(n) \le m-n$.
\item $\forall n, G(n)=0$ if and only if $n=0$.
\item $\forall n>1, G(n)<n$.
\end{enumerate}
\end{theorem}
\begin{proof}
\noindent
\begin{enumerate}
\item Already seen during the definition of $G$.
\item $G(n+1)-G(n) = (n+1)-G(G(n))-n+G(G(n-1)) = 1 - (G(G(n))-G(G(n-1)))$.
If $G(n)=G(n-1)$ then the previous expression is equal to 1.
\item By strong induction over $n$. First, $G(1)-G(0)=1-0$. Then
for a given $n\neq 0$, we assume $\forall k<n, G(k+1)-G(k) \in \{0,1\}$.
We reuse the same formulation of $G(n+1)-G(n)$ as before.
By induction hypothesis for $k=n-1$, $G(n)-G(n-1) \in\{0,1\}$.
If it's 0, then $G(n+1)-G(n) = 1$ as before. If it's 1, we could use another
induction hypothesis for $k=G(n-1)$ (and hence $k+1 = G(n)$), leading
to $G(G(n))-G(G(n-1)) \in \{0,1\}$ and the desired result.
\item Mere iteration of the previous result between $n$ and $m$.
\item For all $n\ge 1$, $0 \le G(n)-G(1)$ and $G(1)=1$.
\item For all $n\ge 2$, $G(n)-G(2) \le n-2$ and $G(2)=1$.
\end{enumerate}
\end{proof}

We can also say that $lim_{+\infty} G = +\infty$, since $G$ is
monotone and grows by at least 1 every 2 steps (any stagnation
is followed by a growth). Taking into account this limit, and
$G(0)=0$, and the growth by steps of 1, we can also deduce
that $G$ is onto: any natural number has at least one antecedent
by $G$. Moreover there cannot exists more than two antecedents
for a given value: by monotonicity, these antecedents are neighbors,
and having more than two would contradict the ``stagnation followed
by growth'' rule.
We can even provide explicitly one of these antecedent:

\begin{theorem}\label{Gonto}
$\forall n, G(n+G(n))=n$.
\end{theorem}
\begin{proof}
We've just proved that $n$ has at least one antecedent, and
no more than two. Let $k$ be the largest of these antecedents.
Hence $G(k)=n$ and $G(k+1)\neq n$, leading to $G(k+1)=G(k)+1$.
If we re-inject this into the defining equation $G(k+1) = k+1 -
G(G(k))$, we obtain that $G(k)+G(G(k))=k$ hence $n+G(n)=k$,
and finally $G(n+G(n))=G(k)=n$.
\end{proof}

As shown during the previous proof, $n+G(n)$ is actually the largest antecedent
of $n$ by $G$. In particular $G(n+G(n)+1)=n+1$. And if $n$ has another
antecedent, it will hence be $n+G(n)-1$.

From this, we can deduce a first relationship between $G$ and
Fibonacci numbers.
\begin{theorem}\label{Gfib} For all $k\ge 2$, $G(F_k)=F_{k-1}$.
\end{theorem}
\begin{proof}By induction over $k$. First, $G(F_2)=G(1)=1=F_1$.
Take now a $k\ge 2$ and assume that $G(F_k)=F_{k-1}$.
Then $G(F_{k+1})=G(F_k+F_{k-1})=G(F_k+G(F_k))=F_k$.
\end{proof}

Moreover, for $k>2$, $F_k = F_{k-1}+G(F_{k-1})$ is the largest antecedent of
$F_{k-1}$, hence $G(1+F_k)=1+F_{k-1}$.

We could also establish a alternative equation for $G$, which will be
used during the study of function $\FG$:

\begin{theorem}\label{Galt} For all $n$ we have $G(n) + G(G(n+1)-1) = n$.
\end{theorem}
\begin{proof}
First, this equation holds when $n=0$ since
$G(0)+G(G(1)-1) = G(0) + G(1-1) = 0$.
We consider now a number $n\neq 0$. Either $G(n+1)=G(n)$ or $G(n+1)=G(n)+1$.
\begin{itemize}
\item In the first case, $G(n-1)$ cannot be equal to $G(n)$ as well,
otherwise $G$ would stay flat longer than possible. Hence $G(n-1)=G(n)-1$.
Hence $G(n)+G(G(n+1)-1) = G(n) + G(G(n)-1) = G(n)+G(G(n-1)) = n$.
\item In the second case:
$G(n) + G(G(n+1)-1) = G(n+1)-1 + G(G(n)) = (n+1)-1 = n$.
\end{itemize}
\end{proof}

\subsection{The associated $G$ tree}

Hofstadter proposed to associate the $G$ function with an infinite
tree satisfying the following properties:
\begin{itemize}
\item The nodes are labeled by all the positive numbers,
 in the order of a left-to-right breadth-first traversal,
 starting at 1 for the root of the tree.
\item For $n\neq 0$, the node $n$ is a child of the node $g(n)$.
\end{itemize}

In practice, the tree can be constructed progressively, node after node.
For instance, the node $2$ is the only child of $1$, and $3$
is the only child of $2$, while $4$ and $5$ are the children of $3$.
Now come the nodes $6$ and $7$ on top of $4$, and so on. The picture
below represents the tree up to depth 7 (assuming that the root is
at depth 0).

\bigskip

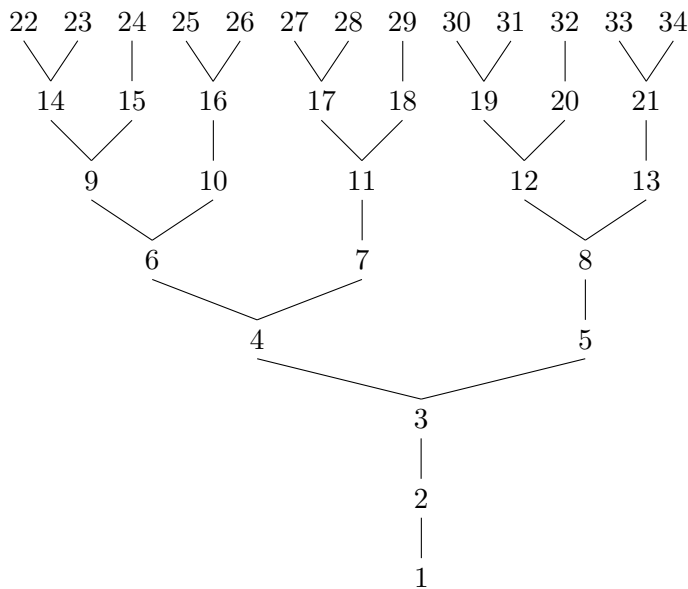
\begin{tikzpicture}[grow'=up]
\Tree
 [.1 [.2 [.3
       [.4 [.6 [.9 [.14 22 23 ] [.15 24 ] ]
               [.10 [.16 25 26 ]]]
           [.7 [.11 [.17 27 28 ] [.18 29 ]]]]
       [.5 [.8 [.12 [.19 30 31 ] [.20 32 ]]
               [.13 [.21 33 34 ]]]]]]]
\end{tikzpicture}

This construction process will always be successful : when adding the
new node $n$, this new node $n$ will be linked
to a parent node already constructed, since $G(n)<n$ as soon as $n>1$.
Moreover, $G$ is monotone and
grows by at most 1 at each step, hence the position of the new node
$n$ will be compatible with the left-to-right breadth-first ordering:
$n$ has either the same parent as $n-1$, and we place $n$ to the right
of $n-1$, or $G(n)=G(n-1)+1$ in which case $n$ is either to be placed
on the right of $n-1$, or at a greater depth.

Moreover, since $G$ is onto, each node will have at least one child,
and we have already seen that each number has at most two antecedents
by $G$, hence the node arities are 1 or 2.

\subsubsection*{Tree depth}
On the previous picture, we can notice that the rightmost nodes are
Fibonacci numbers, while the leftmost nodes have the form $1+F_k$.
Before proving this fact, let us first define more properly a depth
function.

\begin{definition}
For a number $n>0$, we note $\depth(n)$ the least number $d$ such that
$G^d(n)$ (the $d$-th iteration of $G$ on $n$) reaches 1.
We complete this definition by choosing arbitrarily $\depth(0)=0$.
\end{definition}

For $n>0$, such a number $d$ is guaranteed to exists since $G(k)<k$ as long as
$k$ is still not $1$, hence the sequence $G^k(n)$ is strictly
decreasing as long as it hasn't reached 1.
In particular, we have $\depth(1)=0$, and for all $n>1$ we have
$\depth(n) = 1+\depth(G(n))$. Hence our depth function is compatible
with the usual notion of depth in a tree.

\begin{theorem}\label{depthprops}
\noindent
\begin{enumerate}
\item For all $n$, $\depth(n)=0$ if and only if $n\le 1$.
\item For $n>1$ and $k<\depth(n)$, we have $G^k(n) > 1$.
\item For $n,m>0$, $n\le m$ implies $\depth(n)\le \depth(m)$.
\item For $k>1$, $\depth(F_k) = k-2$.
\item For $k>1$, $\depth(1+F_k) = k-1$.
\item For $k>0,n>0$, $\depth(n)=k$ if and only if $1+F_{k+1} \leq n \leq
  F_{k+2}$.
\end{enumerate}
\end{theorem}
\begin{proof}
\noindent
\begin{enumerate}
\item The statement is valid when $n=0$. For a $n>0$, 
$\depth(n)=0$ means $G^0(n)=1$ i.e. $n=1$.
\item Consequence of the minimality of depth.
\item Consequence of the monotony of $G$ :
  for all $k$, $G^k(n)\le G^k(m)$ hence the iterates of $n$ will reach
  1 faster than the iterates of $m$.
\item $G$ maps a Fibonacci number to the previous one, $k-2$
  iterations of $G$ on $F_k$ leads to $F_2=1$.
\item $G$ maps a successor of a Fibonacci number to the previous such
  number, $k-2$ iterations of $G$ on $1+F_k$ leads to $1+F_2 = 2$.
\item If $1+F_{k+1} \leq n \leq F_{k+2}$, the monotony of depth and the
previous facts gives $k \le \depth(n) \le k$, hence $\depth(n)=k$.
Conversely, when $\depth(n)=k$, we cannot have $n < 1 + F_{k+1}$ otherwise
$n \le F_{k+1}$ and hence $\depth(n)\le k-1$, and we cannot have $n > F_{k+2}$
otherwise $n \ge 1+F_{k+2}$ and hence $\depth(n)\ge k+1$.
\end{enumerate}
\end{proof}

The previous characterization of $\depth(n)$ via Fibonacci bounds shows
that depth could also have been defined thanks to {\tt
  fib\_inv$(n-1)-1$}, see \ref{fibinv}. It also shows that the number of
nodes at depth $k$ is $F_{k+2}-(1+F_{k+1})+1 = F_{k}$ as soon as
$k\neq 0$.

\subsubsection*{The shape of the $G$ tree}

If we ignore the node labels and concentrate on the shape of the $G$
tree, we encounter a great regularity. $G$ starts with two unary
nodes that are particular cases (labeled 1 and 2), and after that
we encounter a sub-tree $G'$ whose shape is obtained by repetitions
of the same basic pattern:

\bigskip

G = 
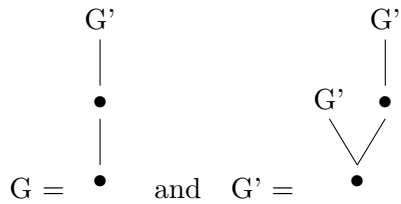
\begin{tikzpicture}[grow'=up]
\Tree [.$\bullet$ [.$\bullet$ G' ]]
\end{tikzpicture}
~~and~~
G' =
\begin{tikzpicture}[grow'=up]
\Tree [.$\bullet$ [.G' ] [.$\bullet$ [.G' ] ]]
\end{tikzpicture}

G' has hence a fractal shape: it appears as a sub-tree of itself. To
prove the existence of such a pattern, we study the arity of children
nodes.
\begin{theorem}\label{Gnodes}
In the $G$ tree, a binary node has a left child which is also binary,
and a right child which is unary, while the unique child of a unary
node is itself binary.
\end{theorem}
\begin{proof}
First, we show that the leftmost child of a node is always binary.
Let $n$ be a node, and $p$ its leftmost child, i.e. $G(p)=n$ and
$G(p-1)=n-1$. We already know one child $q=p+G(p)$ of $p$, let's now
show that $q-1$ is also a child of $p$. For that, we consider $p-1$
and its rightmost child $q'=p-1+G(p-1)$. Rightmost means that
$G(q'+1)=p$. So $G(q-1)=G(p+G(p)-1)=G(p+n-1)=G(p+G(p-1))=G(q'+1)=p$
and we can conclude.

We can now affirm that the unique child of a unary node is itself
binary, since this child is in particular the leftmost one.

Let's now consider a binary node $n$, with its right child $p=n+G(n)$
and its left child $p-1$: $G(p)=G(p-1)=n$.

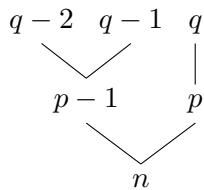
\begin{tikzpicture}[grow'=up]
\Tree [.$n$ [.$p-1$ $q-2$ $q-1$ ] [.$p$ $q$ ]]
\end{tikzpicture}

The leftmost child $p-1$
is already known to be binary. Let's now show that the right child $p$
is unary. If $p$ has a second child, it will be $q-1$, but:
$G(q-1) = G(p+G(p)-1) = G(p-1+G(p-1)) = p-1$.
so $q-1$ is a child of $p-1$ rather than $p$, and $p$ is indeed
a unary node.
\end{proof}

\subsection{$G$ and Fibonacci decompositions}

\begin{theorem}\label{Gshift}
Let $n = \fibrest$ be a relaxed Fibonacci decomposition.
Then $G(n) = \Sigma F_{i-1}$, the Fibonacci decomposition obtained
by removing 1 at all the ranks in the initial decomposition.
\end{theorem}

For instance $G(11) = G(F_4+F_6) = F_3 + F_5 = 7$. Note that the
obtained decomposition isn't necessarily relaxed anymore, since a
$F_1$ might have appeared if $F_2$ was part of the initial
decomposition. The theorem also holds for canonical decompositions
since they are a particular case of relaxed decompositions. Once again,
the obtained decompositions might not be canonical since $F_1$ might
appear, but in this case we could turn this $F_1$ into a $F_2$. This
gives us a relaxed decomposition that we can re-normalize. For
instance $G(12) = G(F_2+F_4+F_6) = F_1+F_3+F_5 = F_2+F_3+F_5 = F_4+F_5
= F_6 = 8$.
 
\begin{proof}
We proceed by strong induction over $n$. The case $0$ is obvious,
since the only possible decomposition is the empty one, and the
``shifted'' decomposition is still empty, as required by $G(0)=0$.
We now consider a number $n\neq 0$, with a non-empty relaxed decomposition
$n = F_k+\fibrest$, and assume the statement to be true for all $m<n$.
We will use the recursive equation $G(n)=n-G(G(n-1))$, and distinguish
many cases according to the value of $k$.
\begin{itemize}
\item Case $k=2$. Then $n-1 = \fibrest$. A first induction
  hypothesis gives $G(n-1) = \Sigma F_{i-1}$. Since the initial
  decomposition was relaxed, $\fibrest$ cannot contain $F_2$,
  hence $\Sigma F_{i-1}$ is still a relaxed decomposition. As seen
  many times now, $G(n-1)<n$ here, and we're free to use a second
  induction hypothesis leading to $G(G(n-1)) = \Sigma F_{i-2}$.
  Finally $G(n) = 1+\Sigma F_i - \Sigma F_{i-2} = 1 + \Sigma (F_i-F_{i-2})
   = F_1 + \Sigma F_{i-1}$, as required.

\item Case $k$ even and distinct from 2. Then $k$ can be written $2p$ with $p>1$.
  Then $n-1 = F_{2p}-1 + \Sigma F_i = F_3+F_5+...+F_{2p-1}+\Sigma F_i$.
  The induction hypothesis on $n-1$ gives:
  $$G(n-1) = F_2+F_4+...+F_{2p-2}+\Sigma F_{i-1}$$
  Since this is still a relaxed decomposition, we use a second
  induction hypothesis on it, hence:
  $$G(G(n-1)) = F_1+F_3+...+F_{2p-3}+\Sigma F_{i-2}$$
  $$G(G(n-1)) = 1 + (F_{2p-2}-1)+\Sigma F_{i-2}$$
  And so:
  $$G(n) = F_{2p}-F_{2p-2} + \Sigma (F_i - F_{i-2}) =
           F_{2p-1} + \Sigma F_{i-1}$$

\item Case $k$ odd. Then $k$ can be written $2p+1$ with $p>0$.
  Then $n-1 = F_{2p+1}-1 + \Sigma F_i = F_2+F_4+...+F_{2p}+\Sigma F_i$.
  The induction hypothesis on $n-1$ gives:
  $$G(n-1) = F_1+F_3+...+F_{2p-1}+\Sigma F_{i-1}$$
  We turn the $F_1$ above into a $F_2$, obtaining again a relaxed
  decomposition, for which a second induction hypothesis gives:
  $$G(G(n-1)) = F_1+F_2+...+F_{2p-2}+\Sigma F_{i-2}$$
  $$G(G(n-1)) = 1 + (F_{2p-1}-1)+\Sigma F_{i-2}$$
  And so:
  $$G(n) = F_{2p+1}-F_{2p-1} + \Sigma (F_i - F_{i-2}) =
           F_{2p} + \Sigma F_{i-1}$$
\end{itemize}
\end{proof}

Thanks to this characterization of the effect of $G$ on Fibonacci
decomposition, we can now derive a few interesting properties.

\begin{theorem}\label{Gclass1}
\noindent
\begin{enumerate}
\item $G(n+1)=G(n)$ if and only if $low(n)=2$.
\item If $low(n)$ is odd, then $G(n-1)=G(n)$.
\item If $low(n)$ is even, then $G(n-1)=G(n)-1$.
\item For $n\neq 0$, the node $n$ of tree $G$
 is unary if and only if $low(n)$ is odd.
\end{enumerate}
\end{theorem}
\begin{proof}
\noindent
\begin{enumerate}
\item
If $low(n)=2$, then we can write a canonical decomposition
$n=F_2+\fibrest$, hence
$n+1=F_3+\fibrest$ is a relaxed decomposition. By the previous
theorem, $G(n) = F_1 + \Sigma F_{i-1}$ while
$G(n+1) = F_2 + \Sigma F_{i-1}$, hence the desired equality.
Conversely, we consider a canonical decomposition of $n$ as $\fibrest$.
If $F_2$ isn't part of this decomposition, then $n+1=F_2+\fibrest$
is a correct relaxed decomposition, leading to
$G(n+1)=1+\Sigma F_{i-1}=1+G(n)$. So $low(n)\neq 2$ implies
$G(n+1)\neq G(n)$.

\item If $low(n)$ is odd, we've already shown in theorem \ref{fibpred}
  that $low(n-1)=2$ hence $G(n)=G(n-1)$ by the previous point.

\item If $low(n)$ is even, the same theorem
 \ref{fibpred} shows that $low(n-1)$ is 3 or more, provided that $n>1$.
 In this case the first point above allows to conclude. We now check separately
 the cases $n=0$ (irrelevant here since $low(0)$ doesn't exists) and
 $n=1$ (for which $G(1-1)$ is indeed $G(1)-1$).

\item We've already seen that a node $n$ is binary whenever 
$G(n+G(n)-1) = n$. When $low(n)$ is odd, we have
$G(n-1)=G(n)$, hence $G(n+G(n)-1) = G(n-1+G(n-1)) = n-1$, hence $n$ is
  unary. When $low(n)$ is even, we have
$n+G(n)-1 = (n-1+G(n-1))+1$ : this is the next node to the right after
the rightmost child of $n-1$, its image by $G$ is hence $n$, and
finally $n$ is indeed binary.
\end{enumerate}
\end{proof}

\begin{theorem}\label{Glow}
If $low(n)>2$, then $low(G(n))=low(n)-1$. In particular,
if $low(n)$ is odd, then $low(G(n))$ is even, and if $low(n)$
is even and different from 2, then $low(G(n))$ is odd.
\end{theorem}
\begin{proof}
This is a direct application of theorem \ref{Gshift}:
when $low(n)>2$, the decomposition we obtain for $G(n)$ is
still canonical, and its lowest rank is $low(n)-1$.
The statements about parity are immediate consequences.
\end{proof}

\begin{theorem}\label{Gclass2}
\noindent
\begin{enumerate}
\item If $low(n)=2$, then $low(G(n))$ is even.
\item If $low(n)=3$, then $low(G(n))=2$.
\item If $low(n)>3$, then $low(G(n))>2$ and $low(G(n)+1)$ is even.
\end{enumerate}
\end{theorem}
\begin{proof}
\noindent
\begin{enumerate}
\item Once again, we consider a canonical decomposition
$n = F_2 + \fibrest$. Then $G(n) = F_1 + \Sigma F_{i-1}$. If we
turn the $F_1$ into a $F_2$, we obtain a relaxed decomposition,
that can be normalized into a canonical decomposition whose
lowest rank will hence be even.
\item Direct application of previous theorem: $low(G(n))=low(n)-1$.
\item Here also, $low(G(n))=low(n)-1$, hence $low(G(n))$ cannot
be 2 here. Then the theorem \ref{fibsucc} implies that $low(G(n)+1)$ is even.
\end{enumerate}
\end{proof}

\begin{theorem}\label{Gthree}
\noindent
\begin{enumerate}
\item If $n$ is 3-even, then $G(n)+1$ is 3-odd.
\item If $n$ is 3-odd, then either $G(n)+1$ is 3-even or $low(G(n)+1)>3$.
\end{enumerate}
\end{theorem}
\begin{proof}
\noindent
\begin{enumerate}
\item Take a canonical decomposition $n=F_3+F_{2p}+\fibrest$.
Hence $G(n)+1=1+F_2+F_{2p-1}+\Sigma F_{i-1}=F_3+F_{2p-1}+\Sigma F_{i-1}$,
and this decomposition is still canonical (for canonicity reasons,
$p>2$).
\item Similarly, $n=F_3+F_{2p+1}+\fibrest$ implies
$G(n)+1=F_3+F_{2p}+\Sigma F_{i-1}$. When $p>1$, this decomposition
is canonical and $G(n)+1$ is 3-even. When $p=1$, this decomposition
is only a relaxed one, starting by $F_3+F_4+...$. Its normalization
will hence end with a lowest rank of at least 5.
\end{enumerate}
\end{proof}

\subsection{$G$ and its ``derivative'' $\Delta G$}
\label{deltaG}

We consider now the ``derivative'' $\Delta G$ of $G$, defined via
$\Delta G(n) = G(n+1)-G(n)$. We already know from theorem
\ref{Gprops} that the output of $\Delta G$ is always either 0 or 1.

\begin{theorem}\label{Gdelta}
For all $n$, $\Delta G(n+1) = 1 - \Delta G(n).\Delta G(G(n))$.
\end{theorem}
\begin{proof}
We already know that $\Delta G(n)=0$ implies $\Delta G(n+1)=1$:
we cannot have $G(n)=G(n+1)=G(n+2)$.
Consider now a $n$ such that
$G(n+1)-G(n)=1$. By using the recursive definition of $G$, we have
$G(n+2)-G(n+1)=(n+2-G(G(n+1)))-(n+1-G(G(n))) = 1 - (G(G(n+1))-G(G(n))$.
Hence $\Delta G(n+1) = 1 - (G(G(n)+1)-G(G(n)) = 1 - \Delta G(G(n))$.
\end{proof}

This equation provides a way to express $G(n+2)$ in terms of
$G(n+1)$ and $G(n)$ and $G(G(n))$ and $G(G(n)+1)$. Since
$n+1$, $n$, $G(n)$ and $G(n)+1$ are all strictly less than $n+2$,
we could use this equation as an alternative way to define
recursively $G$,
alongside two initial equations $G(0)=0$ and $G(1)=1$.
We proved in Coq that $G$ is indeed the unique function to
satisfy these equations (see {\tt GD\_unique} and
{\tt  g\_implements\_GD}).

\section{The $\FG$ function}

This section corresponds to file \doc{FlipG}.
Here is a quote from page 137 of Hofstadter's book \cite{GEB}:
\begin{quote}
A problem for curious readers is: suppose you flip Diagram G
around as if in a mirror, and label the nodes of the new tree so they
increase from left to right. Can you find a recursive \emph{algebraic}
definition for this ``flip-tree'' ?
\end{quote}

\subsection{The $\flip$ function}

Flipping the $G$ tree as if in a mirror is equivalent to keeping
its shape unchanged, but labeling the nodes from right to left
during the breadth-first traversal. Let us call $\flip(n)$ the
new label of node $n$ after this transformation.
We've seen in the previous section that given a depth $k\neq 0$,
the nodes at this depth are labeled from $1+F_{k+1}$ till $F_{k+2}$.
After the $\flip$ transformation, the $1+F_{k+1}$ node and
the $F_{k+2}$ node will hence have exchanged their label.
Similarly, $2+F_{k+1}$ will become $F_{k+2}-1$ and vice-versa.
More generally, if $\depth(n)=k$, the distance between
$\flip(n)$ and the leftmost node $1+F_{k+1}$ will be equal to
the distance between $n$ and the rightmost node $F_{k+2}$,
hence:
$$\flip(n) - (1+F_{k+1}) = F_{k+2} - n$$
So:
$$\flip(n) = 1+F_{k+3}-n$$
We finally complete this definition to handle the case $n\le 1$:
\begin{definition}
We define the function $\flip : \mathbb{N}\to\mathbb{N}$
in the following way:
$$\flip(n) = if~(n\le 1)~then~n~else~1+F_{3+\depth(n)}-n$$
\end{definition}

A few properties of this $\flip$ function:
\begin{theorem}\label{flipprops}
\noindent
\begin{enumerate}
\item For all $n\in\mathbb{N}$, $\depth(\flip(n)) = \depth(n)$.
\item For all $n$, $n>1$ if and only if $\flip(n)>1$.
\item For $1 \le n \le F_{k}$, we have
 $\flip(F_{k+1}+n) = 1+F_{k+2}-n$.
\item $\flip$ is involutive.
\item For all $n>1$, if $\depth(n+1)=\depth(n)$ then
  $\flip(n+1)=\flip(n)-1$.
\item For all $n>1$, if $\depth(n-1)=\depth(n)$ then
  $\flip(n-1)=\flip(n)+1$.
\end{enumerate}
\end{theorem}
\begin{proof}
\noindent
\begin{enumerate}
\item If $n\le 1$, then $\flip(n)=n$ and the property is obvious.
For $n>1$, if we name $k$ the depth of $n$, we have
$\flip(n) = 1+F_{k+3}-n$. Since $1+F_{k+1} \le n \le F_{k+2}$,
we hence have $1+F_{k+3}-F_{k+2} \le \flip(n) \le 1+F_{k+3}-1-F_{k+1}$.
And finally $1+F_{k+1} \le \flip(n) \le F_{k+2}$, and this characterizes
the nodes at depth $k$, hence $\depth(\flip(n))=k=\depth(n)$.
\item We know that $n$ and $\flip(n)$ have the same depth, and
we've already seen that being less or equal to 1 is equivalent
to having depth 0.
\item Consider $1 \le n \le F_{k}$. Hence $k$ is at least 1,
otherwise $1 \le F_0 = 0$. We have
$1+F_{k+1}\le F_{k+1}+n \le F_{k+1}+F_{k} = F_{k+2}$, so the depth
of $F_{k+1}+n$ is $k$
(still via the same characterization of nodes at depth $k$).
Moreover $F_{k+1}+n$ is more than 2,
so the definition of $\flip$ gives:
$\flip(F_{k+1}+n)=1+F_{k+3}-F_{k+1}-n = 1+F_{k+2}-n$.
\item If $n\le 1$ then $\flip(n)=n$ is still less or equal to 1,
so $\flip(\flip(n))=\flip(n)=n$.
Consider now $n>1$. We already know that $\flip(n)>1$ and
$\flip(n)$ and $n$ have same depth (let us name it $k$). Hence:
$\flip(\flip(n)) = 1+F_{k+2}-\flip(n)=1+F_{k+2}-(1+F_{2+k}-n) = n$.
\item Let us name $k$ the common depth of $n+1$ and $n$.
In these conditions
$\flip(n+1) = 1+F_{k+2}-(n+1) = (1+F_{k+2}-n)-1 = \flip(n)-1$.
\item Similar proof.
\end{enumerate}
\end{proof}

In particular, the third point above shows that for $k>1$
we indeed have $\flip(1+F_k) = F_{k+1}$ (and vice-versa since $\flip$
is involutive).

\subsection{Definition of $\FG$ and initial properties}

We can now take advantage of this $\flip$ function to obtain
a first definition of the $\FG$ function, corresponding to the
flipped $G$ tree:

\begin{definition}
The function $\FG : \mathbb{N}\to\mathbb{N}$ is defined via:
$$\FG(n) = \flip(G(\flip(n))$$
\end{definition}

We benefit from the involutive aspect of $\flip$ to switch
from right-left to left-right diagrams, use $G$, and then
switch back to right-left diagram. The corresponding Coq function
is named {\tt fg}.

By following this definition, the initial values of $\FG$ are
$\FG(0)=0$, $\FG(1)=\FG(2)=1$, $\FG(3)=2$. The first difference
between $\FG$ and $G$ appears for $\FG(7)=5=G(7)+1$. We'll dedicate a
whole section later on the comparison between $\FG$ and $G$.
Here is the initial levels of this flipped tree $\FG$,
the boxed values being the places where $\FG$ and $G$ differ.

\bigskip

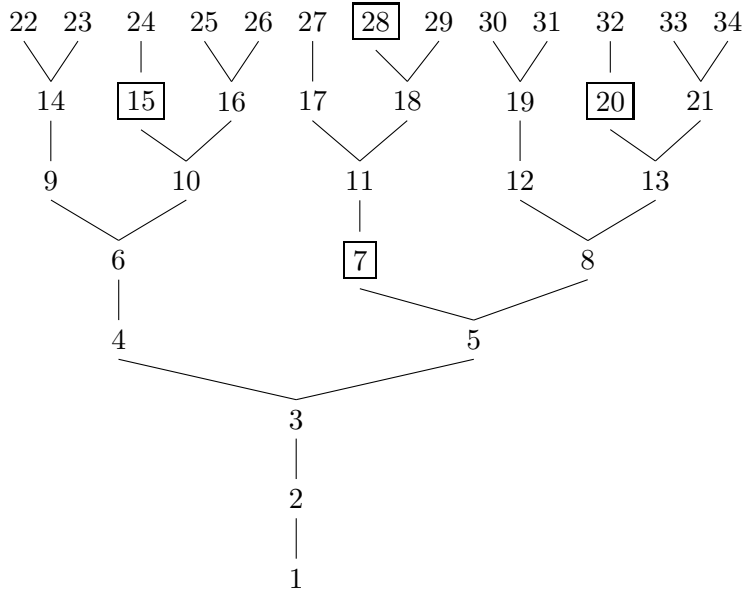
\begin{tikzpicture}[grow'=up]
\Tree
 [.1 [.2 [.3
       [.4 [.6 [.9 [.14 22 23 ] ]
               [.10 [.\fbox{15} 24 ] 
                    [.16 25 26 ]]]]
       [.5 [.\fbox{7} [.11 [.17 27 ] [.18 \fbox{28} 29 ]]]
           [.8 [.12 [.19 30 31 ] ]
               [.13 [.\fbox{20} 32 ]
                    [.21 33 34 ]]]]]]]
\end{tikzpicture}

We now show that $\FG$ enjoys the same basic properties as $G$:

\begin{theorem}\label{FGprops}
\noindent
\begin{enumerate}
\item For all $n>1$, $\depth(\FG(n)) = \depth(n)-1$.
\item For all $k>0$, $\FG(F_{k+1}) = F_k$.
\item For all $k>1$, $\FG(1+F_{k+1}) = 1+F_k$.
\item For all $n$, $\FG(n+1)-\FG(n) \in \{0,1\}$.
\item For all $n,m$, $n\le m$ implies $0 \le \FG(m)-\FG(n) \le m-n$.
\item For all $n$, $0 \le \FG(n) \le n$.
\item For all $n$, $\FG(n)=0$ if and only if $n=0$
\item For all $n>1$, $\FG(n)<n$.
\item For all $n\neq 0, \FG(n)=\FG(n-1)$ implies $\FG(n+1)=\FG(n)+1$.
\item $\FG$ is onto.
\end{enumerate}
\end{theorem}
\begin{proof}
\noindent
\begin{enumerate}
\item Since $\flip$ doesn't modify the depth, we reuse the
result about the depth of $G(n)$.
\item $\FG(F_{k+1}) = \flip(G(\flip(F_{k+1}))) = \flip(G(1+F_k)) =
       \flip(1+F_{k-1}) = F_k$.
\item Similar proof.
\item The property is true for $n=0$ and $n=1$.
Consider now $n>1$. Let $k$ be $\depth(n)$, which is
hence different from 0. We have $1+F_{k+1} \le n \le F_{k+2}$.
If $n$ is $F_{k+2}$, we have already seen that
$\FG(1+F_{k+2}) = 1+F_{k+1} = 1+\FG(F_{k+2})$. Otherwise $n < F_{k+2}$
and $n+1$ has hence the same depth as $n$. Then
$\flip(n+1)=\flip(n)-1$. Now, by applying $G$ on the two consecutive values
of same depth $\flip(n)$ and $\flip(n)-1$, we obtain two results that
are either equal or consecutive, and of same depth. By applying $\flip$
again on these $G$ results, we end on two consecutive or equal
$\FG$ results.
\item Iteration of the previous results between $n$ and $m$.
\item Thanks to the previous point, $0 \le \FG(m)-\FG(0) \le m-0$
and $\FG(0)=0$.
\item $\FG(1)=1$ and $0 \le \FG(m)-\FG(1)$ as soon as $m\ge 1$.
\item $\FG(2)=1$ and $\FG(m)-\FG(2) \le m-2$ as soon as $m \ge 2$.
\item If $n=1$, then the precondition $\FG(1)=\FG(0)$ isn't
  satisfied. If $n=2$ or $n=3$, the conclusion is satisfied:
  $\FG(3)=2=\FG(2)+1$ and $\FG(4)=3=\FG(3)+1$.
  We now consider $n>3$ such that
  $\FG(n)=\FG(n-1)$. We have already seen that $\FG(n+1)-\FG(n)$
  is either 0 or 1. If it is 1, we can conclude. We proceed by
  contradiction and suppose it is 0. We hence have
  $\FG(n-1)=\FG(n)=\FG(n+1)$, or said otherwise
  $\flip(G(\flip(n-1))) = \flip(G(\flip(n))) = \flip(G(\flip(n+1)))$.
  Since $\flip$ is involutive hence bijective, this implies that
  $G(\flip(n-1))=G(\flip(n))=G(\flip(n+1))$. To be able to commute
  $\flip$ and the successor/predecessor, we now study the depth
  of these values.
  \begin{itemize}
  \item If $n+1$ has a different depth than $n$, then
    $n$ is a Fibonacci number $F_{k}$, with $k>4$ since $n>3$.
    Then $\flip(n+1)=F_{k+1}$ and $\flip(n)=1+F_{k-1}$. These values
    aren't equal nor consecutive, hence $G$ cannot give the same
    result for them, this situation is indeed contradictory.
  \item Similarly, if $n-1$ has a different depth than $n$, then
    $n-1$ is a Fibonacci number $F_k$, with $k>3$ since $n>3$.
    Then $\flip(n-1)=1+F_{k-1}$ while $\flip(n)=F_{k+1}$.
    Once again, these values aren't equal nor consecutive,
    hence $G$ cannot give the same result for them. Contradiction.
  \item In the last remaining case, 
    $\depth(n-1)=\depth(n)=\depth(n+1)$.
    Then $\flip(n+1)=\flip(n)-1$ and $\flip(n-1)=\flip(n)+1$ and
    these value are at distance 2, while having the
    same result by $G$, which is contradictory.
  \end{itemize}
\item We reason as we did earlier for $G$: $\FG$ starts with
$\FG(0)=0$, it grows by steps of 0 or 1, and it grows by at
least 1 every two steps. Hence its limit is $+\infty$ and it is
an onto function.
\end{enumerate}
\end{proof}

\subsection{An recursive algebraic definition for $\FG$}

The following result is an answer to Hofstadter's problem.
It was already mentioned on OEIS page \cite{OEIS-FG}, but only as
a conjecture.

\newcommand{\nn}{\overline{n}}

\begin{theorem}\label{FGeqn}
For all $n>3$ we have $\FG(n) = n+1 - \FG(1+\FG(n-1))$.
And $\FG$ is uniquely characterized by this equation
plus initial equations $\FG(0)=0$, $\FG(1)=\FG(2)=1$ and
$\FG(3)=2$.
\end{theorem}
\begin{proof}
We consider $n>3$. Let $k$ be its depth, which is hence at
least 3. We know that $1+F_{k+1} \le n \le F_{k+2}$. We will note
$\flip(n)$ below as $\nn$. Roughly speaking, this proof is
essentially a use of theorem \ref{Galt} which implies that
$G(G(\nn+1)-1) = \nn - G(\nn)$, and then some
play with flip and predecessors and successors, when possible,
or direct particular proofs otherwise.
\begin{itemize}
\item We start with the most general case : we suppose here
$\depth(n-1)=\depth(n)$ and $\depth(\FG(n-1)+1)=\depth(\FG(n-1)$.
Let us shorten $\FG(n-1)+1$ as $p$. In particular
$\depth(p)=\depth(\FG(n-1))=\depth(n-1)-1=\depth(n)-1=k-1$.
Due to the equalities between depths, and the facts that
$n>1$ and $\FG(n-1)>1$, we're allowed to use the
last properties of theorem \ref{flipprops} : $\flip(n-1)=\nn+1$ and
\begin{align*}
\flip(p) &= \flip(\FG(n-1)+1) \\
        &= \flip(\FG(n-1))-1 \\
        &= \flip(\flip(G(\flip(n-1))))-1 \\
        &= G(\flip(n-1))-1 \\
        &= G(\nn+1)-1
\end{align*}
We now exploit the definition of $\flip$ three times:
\begin{enumerate}
\item $\depth(n)=k\neq 0$ implies $\nn=1+F_{k+3}-n$
\item $\depth(G(\flip(p)))=\depth(p)-1=k-2 \neq 0$ hence:
$$\FG(p) = \flip(G(\flip(p))) = 1+F_{k+1}-G(\flip(p))$$
\item $\depth(G(\nn))=\depth(\nn)-1=\depth(n)-1=k-1 \neq 0$, so:
$$\FG(n)=\flip(G(\nn))=1+F_{k+2}-G(\nn)$$
\end{enumerate}
All in all:
\begin{align*}
n+1-\FG(1+\FG(n-1)) & = n+1-\FG(p) \\
                    & = n+1-(1+F_{k+1}-G(\flip(p))) \\
                    & = n-F_{k+1}+G(G(\nn+1)-1)) \\
                    & = n-F_{k+1}+(\nn-G(\nn)) \\
                    & = n-F_{k+1}+(1+F_{k+3}-n)-G(\nn) \\
                    & = 1+F_{k+2}-G(\nn) \\
                    & = \FG(n)
\end{align*}
\item We now consider the case where $\depth(n-1)\neq \depth(n)$.
So $n$ is the least number of depth $k$, hence $n=1+F_{k+1}$.
In this case:
\begin{align*}
1+n-\FG(1+\FG(n-1)) & = 2+F_{k+1}-\FG(1+\FG(F_{k+1})) \\
                    & = 2+F_{k+1}-\FG(1+F_{k}) \\
                    & = 2+F_{k+1}-(1+F_{k-1}) \\
                    & = 1+F_{k} \\
                    & = \FG(n)
\end{align*}
\item The last case to consider is $\depth(n-1)=\depth(n)$ but
$\depth(\FG(n-1)+1)\neq \depth(\FG(n-1))$. As earlier, we prove
that $\depth(\FG(n-1))=\depth(n-1)-1=\depth(n)-1=k-1$. So
$\FG(n-1)$ is the greatest number of depth $k-1$, hence
$\FG(n-1)=F_{k+1}$. Since $\FG(F_{k+2})$ is also $F_{k+1}$, then $n-1$
and $F_{k+2}$ are at a distance of 0 or 1. But we know that
$\depth(n)=k$, hence $n \le F_{k+2}$, so $n-1 < F_{k+2}$.
Finally $n=F_{k+2}$ and:
\begin{align*}
1+n-\FG(1+\FG(n-1)) & = 1+F_{k+2}-\FG(1+F_{k+1}) \\
                    & = 1+F_{k+2}-(1+F_{k}) \\
                    & = F_{k+1} \\
                    & = \FG(n)
\end{align*}
\end{itemize}

Finally, if we consider another function $F$ satisfying
the same recursive equation, as well as the same initial values
for $n\le 3$, then we prove by strong induction over $n$ that
$\forall n, F(n)=\FG(n)$. This is clear for $n\le 3$. Consider now
some $n>3$, and assume that $F(k)=\FG(k)$ for all $k<n$.
In particular $F(n-1)=\FG(n-1)$ by induction hypothesis for
$n-1<n$.
Moreover $n-1>1$ hence $\FG(n-1) < n-1$ hence $\FG(n-1)+1<n$.
So we could use a second induction hypothesis at this position:
$F(\FG(n-1)+1)=\FG(\FG(n-1)+1)$. This combined with the first
induction hypothesis above gives $F(F(n-1)+1)=\FG(\FG(n-1)+1)$
and finally $F(n)=\FG(n)$ thanks to the recursive equations
for $F$ and $\FG$.
\end{proof}

\subsection{The $\FG$ tree}

Since $\FG$ has been obtained as the mirror of $G$, we already
know its shape: it's the mirror of the shape of $G$. We'll
nonetheless be slightly more precise here. The proofs given
in this section will be deliberately sketchy, please consult
theorem {\tt unary\_flip} and alii in file \doclab{FlipG}{unary\_flip}\ for
more detailed and rigorous justifications.

First, a node
$n$ has a child $p$ in the tree $\FG$ is and only if the
node $\flip(n)$ has a child $\flip(p)$ in the tree $G$.
Indeed, $\FG(p)=n$ iff $\flip(G(\flip(p)))=n$ iff
$G(\flip(p))=\flip(n)$. Moreover, a rightmost child in $\FG$
corresponds by $\flip$ with a leftmost child in $G$ and vice-versa.
Indeed, if $p$ and $n$ aren't on the borders of the trees,
then the next node will become the previous by flip, and vice-versa
(see theorem \ref{flipprops}). And if $p$ and/or $n$ are on the
border, they are Fibonacci numbers or successors of Fibonacci
numbers, and we check these cases directly.

Similarly, the arity of $n$ in $\FG$ is equal to the arity of
$\flip(n)$ in $G$.
For instance, if we take a unary node $n$ and its unique
child $p$ in $\FG$, this means that $\FG(p+1)=n+1$ and
$\FG(p-1)=n-1$. In the most general case $\flip$ will lead
to similar properties about $\flip(p)$ and $\flip(n)$ in $G$,
hence the fact that $\flip(n)$ is unary in $G$. And we handle
particular cases about Fibonacci numbers on the border as usual.

\begin{theorem}
For all $n>1$, $\flip(\flip(n)+G(\flip(n)))$, which could also be
written $\flip(\flip(n)+\flip(\FG(n)))$, is the leftmost child
of $n$ in the $\FG$ tree. 
\end{theorem}
\begin{proof}
See the previous paragraph: $\flip(n)+G(\flip(n))$ is rightmost
child of $\flip(n)$ in $G$, hence the result about $\FG$.
\end{proof}

\begin{theorem}
For all $n>1$, $n-1+\FG(n+1)$ is the rightmost child of $n$ in the
$\FG$ tree.
\end{theorem}
\begin{proof}
We already know that all nodes $n$ in the $\FG$ tree have at
least one antecedent, and no more than two.
Take $n>1$, and let $k$ be the largest of its antecedents by $\FG$.
Hence $\FG(k)=n$ and $\FG(k+1)\neq n$, leading to $\FG(k+1)=n+1$.
If we re-inject this into the previous recursive equation
$\FG(k+1) = k+2 - \FG(\FG(k)+1)$,
we obtain that $n+1 = k+2-\FG(n+1)$ hence $k=n-1+\FG(n+1)$.
\end{proof}

Of course, for unary nodes, there is only one child, hence the
leftmost and rightmost children given above coincide. Otherwise,
for binary nodes, they are apart by 1.

\begin{theorem}
In the $\FG$ tree, a binary node has a right child which is also binary,
and a left child which is unary, while the unique child of a unary
node is itself binary.
\end{theorem}
\begin{proof}
Flipped version of theorem \ref{Gnodes}.
\end{proof}

\subsection{Comparison between $\FG$ and $G$}

We already noticed earlier that $\FG$ and $G$ produce very similar
answers. Let us study this property closely now.

\begin{theorem}\label{comp-fg-g}
For all $n$, we have $\FG(n)=1+G(n)$ whenever $n$ is 3-odd,
and $\FG(n)=G(n)$ otherwise.
\end{theorem}
\begin{proof}
We proceed by strong induction over $n$.
When $n\le 3$, $n$ is never \mbox{3-odd}, and we indeed have
$\FG(n)=G(n)$. We now consider $n>3$, and assume that
the result is true at all positions strictly less than $n$.
For comparing $\FG(n)$ and $G(n)$ we use their corresponding
recursive equations: when $n$ is 3-odd we need to prove that
$\FG(\FG(n-1)+1)=G(G(n-1))$, and otherwise we need to establish
that $\FG(\FG(n-1)+1)=1+G(G(n-1))$. So we'll need
induction hypotheses (IH) for $n-1<n$ and for $\FG(n-1)+1$
(which is indeed strictly less than $n$:
since $1<n-1$, we have $\FG(n-1)<n-1$). But the exact equations
given by these two IH will depend on the
status of $n-1$ and $\FG(n-1)+1$ : are these numbers 3-odd or not ?
For determining that, we'll consider the Fibonacci decomposition
of $n$ and study its classification.
\begin{itemize}
\item Case $low(n)=2$. Hence $n$ isn't 3-odd and we try to prove
  $\FG(\FG(n-1)+1)=1+G(G(n-1))$. Theorem \ref{fibpred} implies that
  $low(n-1)>3$, and in particular $n-1$ isn't 3-odd.
  So the first IH is $\FG(n-1)=G(n-1)$. By theorem
  \ref{Gclass1}, we also know that $G(n-1)=G(n)-1$.
  So $\FG(n-1)+1=G(n-1)+1=G(n)$
  By theorem \ref{Gclass2}, we know that $G(n)$ has an even lowest
  rank, so it cannot be 3-odd, and the second IH is
  $\FG(\FG(n-1)+1)=G(\FG(n-1)+1)$.
  Since $low(G(n))$ is even, we can use theorem \ref{Gclass1} once
  more: $G(G(n)-1) = G(G(n))-1$.
  Finally: $\FG(\FG(n-1)+1)=G(\FG(n-1)+1)=G(G(n))=1+G(G(n-1))$.
\item Case $n$ 3-odd. We're trying to prove
  $\FG(\FG(n-1)+1)=G(G(n-1))$. Theorem \ref{fibpred} implies that
  $low(n-1) = 2$, and in particular $n-1$ isn't 3-odd.
  So the first IH is $\FG(n-1)=G(n-1)$.
  By theorem \ref{Gclass1}, we also know that $G(n-1)=G(n)$.
  Moreover theorem \ref{Gthree} shows that $\FG(n-1)+1 = G(n)+1$
  cannot be 3-odd.
  So the second IH gives $\FG(\FG(n-1)+1) = G(\FG(n-1)+1)$.
  Now, $low(G(n))=2$ by theorem \ref{Gclass2} so $G(G(n)+1)=G(G(n)$
  by theorem \ref{Gclass1}, or equivalently $G(\FG(n-1)+1)=G(G(n-1))$,
  hence the desired equation.
\item Case $n$ 3-even. As in all the remaining cases, we're now trying here
  to prove $\FG(\FG(n-1)+1)=1+G(G(n-1))$.
  This case is very similar to the previous one until the point
  where we study $G(n)+1$ which is now 3-odd (still thanks to theorem
  \ref{Gthree}).
  So the second IH gives now
  $\FG(\FG(n-1)+1) = 1+G(\FG(n-1)+1)$. And we conclude just as before
  by ensuring for the same reasons that $G(\FG(n-1)+1)=G(G(n-1))$.
\item Case $low(n)=4$.
  Since 4 is even, we also have $G(n-1)=G(n)-1$.
  Hence $low(G(n))=3$ by theorem \ref{Glow}, and $G(G(n)-1)=G(G(n))$
  and $G(G(n)+1)=1+G(G(n))$, both by theorem \ref{Gclass1}.
  Finally $1+G(G(n-1)) = 1+G(G(n)) = G(G(n)+1)$.
  Now, to determine whether $n-1$ is 3-odd,
  we need to look deeper in the canonical decomposition of
  $n = F_4+F_k+\fibrest$.
  \begin{itemize}
  \item If the second lowest rank $k$ is odd, then
    $n-1 = F_3+F_k+\fibrest$ is hence 3-odd, and the first IH is
    $\FG(n-1)=1+G(n-1)$. So
    $\FG(n-1)+1=G(n)+1$, and this number has an even lowest rank
    (theorem \ref{Gclass2}), it cannot be 3-odd, and the second IH
    is:
    $\FG(\FG(n-1)+1)=G(\FG(n-1)+1)$. And this is known to be equal
    to $G(G(n)+1)=1+G(G(n-1))$.
  \item Otherwise $k$ is even and $n-1$ is 3-even
    and hence not 3-odd, so the first IH is $\FG(n-1)=G(n-1)$.
    Moreover, since $n-1$ is 3-even then $\FG(n-1)+1=G(n-1)+1$ is 3-odd by
    theorem \ref{Gthree}. So the second IH is:
    $\FG(\FG(n-1)+1)=1+G(\FG(n-1)+1)$. And this is known to be equal
    to $1+G(G(n-1)+1)=1+G(G(n))=1+G(G(n-1))$.
  \end{itemize}
\item Case $low(n)>4$ and odd.
  By theorem \ref{fibpred}, $low(n-1)=2$ hence the first IH is
  $\FG(n-1)=G(n-1)$. We also know that $G(n-1)=G(n)$ by theorem
  \ref{Gclass1}, and that $low(\FG(n-1)+1)=low(G(n)+1)$ is even
  by theorem \ref{Gclass2}. The second IH is hence:
  $\FG(\FG(n-1)+1) = G(\FG(n-1)+1)$. Finally
  $G(G(n)+1)=1+G(G(n))$ since $low(G(n))=2\neq 1$.
\item Case $low(n)>4$ and even.
  In this case, $n-1$ is 3-odd (theorem \ref{threeevenodd}), so
  the first IH is $\FG(n-1)=1+G(n-1)$. We also know that
  $G(n-1)=G(n)-1$ by theorem \ref{Gclass1}.
  By theorem \ref{Glow} we know that $low(G(n))=low(n)-1 > 3$, so
  by theorem \ref{fibsucc} $low(G(n)+1)=2$, and
  $\FG(n-1)+1=1+G(n)$ cannot be 3-odd : the second IH is hence
  $\FG(\FG(n-1)+1)=G(\FG(n-1)+1)$. Moreover $low(G(n))$ is
  also known to be odd, so 
  by theorem \ref{Gclass1} we have $G(G(n-1))=G(G(n)-1)=G(G(n))$.
  The final step is
  $G(1+G(n))=1+G(G(n))$ also by theorem \ref{Gclass1} (since $low(G(n))\neq
  2$).
\end{itemize}
\end{proof}

As an immediate consequence, $\FG$ is always greater or equal
than $G$, but never more than $G+1$. And we've already studied
the distance between 3-odd numbers, which is always 5 or 8,
while the first 3-odd number is 7. So $\FG$ and $G$ are actually
equal more than 80\% of the time.

\subsection{$\FG$ and its ``derivative'' $\Delta\FG$}

We consider now the ``derivative'' $\Delta\FG$ of $\FG$, defined via
$\Delta \FG(n) = \FG(n+1)-\FG(n)$. This study will be quite
similar to the corresponding section \ref{deltaG} for $G$.
We already know from theorem
\ref{FGprops} that the output of $\Delta\FG$ is always either 0 or 1.

\begin{theorem}\label{FGdelta}
  For all $n>2$, $\Delta\FG(n+1) = 1 - \Delta\FG(n).\Delta\FG(\FG(n+1))$.
\end{theorem}
\begin{proof}
We already know that $\Delta\FG(n)=0$ implies $\Delta\FG(n+1)=1$:
we cannot have $\FG(n)=\FG(n+1)=\FG(n+2)$.
Consider now some $n>2$ such that
$\FG(n+1)-\FG(n)=1$. By using the recursive equation of $\FG$
for $n+1>3$ and $n+2>3$, we have as expected
\begin{align*}
 \Delta\FG(n+1) & = \FG(n+2)-\FG(n+1) \\
                & =(n+3-\FG(\FG(n+1)+1))-(n+2-\FG(\FG(n)+1)) \\
                & = 1 - (\FG(\FG(n+1)+1)-\FG(\FG(n)+1) \\
                & = 1 - (\FG(\FG(n+1)+1)-\FG(\FG(n+1)) \\
                & = 1 - \Delta\FG(\FG(n+1))
\end{align*}
\end{proof}

Note: when compared with theorem \ref{Gdelta} about $\Delta G$, 
the equation
above looks really similar, but has a different inner call
($\Delta\FG(\FG(n+1))$ instead of $\Delta G(G(n))$), and doesn't
hold for $n=2$.

For $n>2$, this equation provides a way to express $\FG(n+2)$ in terms of
$\FG(n+1)$ and $\FG(n)$ and $\FG(\FG(n+1))$ and $\FG(\FG(n+1)+1)$.
For $n>1$, we know that $G(n+1)<n+1$ hence $\FG(n+1)+1<n+2$.
It is also clear that $n+1$ and $n$ and $\FG(n+1)$ are all
strictly less than $n+2$.
So we could use this equation as an alternative way to define
recursively $\FG$,
alongside  initial equations $\FG(0)=0$ and $\FG(1)=\FG(2)=1$ and
$\FG(3)=2$ and $\FG(4)=3$.
We proved in Coq that $\FG$ is indeed the unique function to
satisfy these equations (see {\tt FD\_unique} and
{\tt  fg\_implements\_FD}).

\subsection{An alternative recursive equation for $\FG$}

During the search for an algebraic definition of $\FG$, we
first discovered and proved correct the following result,
which unfortunately doesn't uniquely characterize $\FG$.
We mention it here nonetheless, for completeness sake.

\begin{theorem}
For all $n>3$ we have $\FG(n-1) +\FG(\FG(n)) = n$.
\end{theorem}
\begin{proof}
We proceed as for theorem \ref{FGeqn}, except that we
use internally the original recursive equation for $G$
instead of the alternative equation of theorem \ref{Galt}.
We take $n>3$, call $k$ its depth (which is greater or equal to 3),
and note $\nn = \flip(n)$. First:
$$\FG(\FG(n))=\flip(G(\flip(\flip(G(\nn)))))=\flip(G(G(\nn)))$$
The depth of $G(G(\nn))$ is $k-2\neq 0$, so the definition
of $\flip$ gives:
$$\FG(\FG(n)) = \flip(G(G(\nn))) = 1+F_{k+1}-G(G(\nn))$$
Now, thanks to the recursive definition of $G$ at $\nn+1$, we have
$$\FG(\FG(n)) = 1+F_{k+1}+G(\nn+1)-\nn-1$$
\begin{itemize}
\item If $\depth(n-1)=k$ then $\flip(n-1)=\nn+1$ and we
conclude by more uses of the $\flip$ definition:
$\depth(G(\nn+1))=\depth(G(\flip(n-1))=\depth(n-1)-1=k-1$, so:
$$\FG(n-1)=\flip(G(\flip(n-1)))=\flip(G(\nn+1))= 1+F_{k+2}-G(\nn+1)$$
And finally:
$$\FG(\FG(n)) = F_{k+1} + (1+F_{k+2}-\FG(n-1))-(1+F_{k+3}-n) = n-\FG(n-1)$$

\item If $\depth(n-1)\neq k$ then $n$ is the least number at
depth $k$, so $n=1+F_{k+1}$, and:
$$\FG(n-1) + \FG(\FG(n)) = \FG(F_{k+1})+\FG(\FG(1+F_{k+1})) =
F_{k}+(1+F_{k-1}) = n$$
\end{itemize}
\end{proof}

Note that the following function $f:\mathbb{N}\to\mathbb{N}$
also satisfies this equation, even with the same initial
values than $\FG$:
\begin{align*}
   f(0) & = 0 \\
   f(1) & = f(2)=1  \\
   f(3) & = 2 \\
   f(4) & =f(5) = 3 \\
   f(6) & = 5 \\
   f(7) & = 3 \\
 \forall n\le 4, ~ f(2n) & = n-2 \\
 \forall n\le 4, ~ f(2n+1) & = 4
\end{align*}
This function isn't monotone. We actually proved in Coq that
$\FG$ is the only monotone function that satisfies the previous
equation and the initial constraints $0\mapsto 0$, $1\mapsto 1$, $2\mapsto 1$, $3\mapsto
2$ (see {\tt alt\_mono\_unique} and {\tt alt\_mono\_is\_fg}).
We will not detail these proofs here, the key ingredient is to
prove that any monotone function satisfying these equations
will grow by at most 1 at a time.

\section{Conclusion}

The proofs for theorems \ref{FGeqn} and \ref{comp-fg-g} are
surprisingly tricky, all our attempts at simplifying them
have been rather unsuccessful. But there's probably still room
for improvements here, please let us know if you find or
encounter nicer proofs.

Perhaps using the definition of $G$
via real numbers could help shortening some proofs.
We proved in file \doc{Phi} this definition
$\forall n, G(n)=\lfloor (n+1)/\varphi\rfloor$ where $\varphi$
is the golden ratio $(1+\sqrt{5})/2$. But this has been done
quite late in our development, and we haven't tried to use this
fact for earlier proofs. Anyway,
relating this definition with the flipped function $\FG$
doesn't seem obvious.

Another approach might be to relate
$\FG$ more directly to some kind of Fibonacci decomposition.
Of course, now that theorem \ref{comp-fg-g} is proved, we
know that $\FG$ shifts the ranks of the Fibonacci decompositions
just as $G$, except for 3-odd numbers where a small $+1$
correction is needed. But could this fact be established more
directly ? For the moment, we are only aware of the following
``easy'' formulation of $\FG$ via decompositions:
if $n$ is written as $F_{k+2}-\fibrest$ where $k=\depth(n)$
and $\fibrest$ form a canonical decomposition of $F_{k+2}-n$, then
$\FG(n)=F_{k+1}-\Sigma F_{i-1}+\epsilon$, where $\epsilon=1$ whenever
the decomposition above includes $F_2$, and $\epsilon=0$
otherwise. But this statement didn't brought any new insights
for the proof of theorem \ref{comp-fg-g}, so we haven't formulated
it in Coq.

As possible extensions of this work, we might consider later
the other recursive functions proposed by Hofstadter :
\begin{itemize}
\item $H$ defined via $H(n)=n-H(H(H(n-1)))$, or its generalized
   version with an arbitrary number of sub-calls instead of 2
   for $G$ and 3 for $H$.
\item $\overline{H}$, the flipped version of $H$.
\item $M$ and $F$, the mutually recursive functions (``Male''
and ``Female'').
\end{itemize}
We've already done on paper a large part of the analysis of
$M$ and $F$, and should simply take the time to certify
it in Coq. The study of $H$ and $\overline{H}$ remains to be
done, it looks like a direct generalization of what we've
done here, but surprises are always possible.

\end{document}